\documentclass{article}
\usepackage{amsfonts}
\usepackage{amsmath}
\usepackage{graphicx}
\usepackage{authblk}
\usepackage{hyperref}
\usepackage[square,sort,comma,numbers]{natbib}

\newtheorem{theorem}{Theorem}

\newenvironment{proof}[1][Proof]{\noindent\textbf{#1.} }{\ \rule{0.5em}{0.5em}}

\title{Calculation of the Relaxation Modulus in the Andrade Model by Using the Laplace Transform}
\author{Juan Luis Gonz\'{a}lez-Santander $^{1,}$, Giorgio Spada $^{2}$, Francesco Mainardi $^{3}$ and Alexander Apelblat $^{4}$}
\affil{%
$^{1}$ \quad Department de Mathematics, University of Oviedo, C/ Leopoldo Calvo Sotelo 18, 33007 Oviedo, Spain\\
$^{2}$ \quad Department of Physics and Astronomy Augusto Righi, University of Bologna, Viale Berti Pichat 8, \linebreak I-40127 Bologna, Italy; giorgio.spada@unibo.it\\
$^{3}$ \quad Department of Physics and Astronomy, University of Bologna, and INFN, Via Irnerio 46,\linebreak  I-40126 Bologna, Italy; francesco.mainardi@bo.infn.it\\
$^{4}$ \quad Department of Chemical Engineering, Ben Gurion University of the Negev, Beer Sheva 84105, Israel; apelblat@bgu.ac.il}

\begin{document}

\maketitle

\begin{abstract}
In the framework of the theory of linear viscoelasticity, we derive an
analytical expression of the relaxation modulus in the Andrade model $G_{\alpha }\left( t\right) $ for the case of rational parameter \mbox{$\alpha
=m/n\in (0,1)$} in terms of Mittag--Leffler functions from its Laplace
transform $\tilde{G}_{\alpha }\left( s\right) $. It turns out that the
expression obtained can be rewritten in terms of Rabotnov functions.
Moreover, for the original parameter $\alpha =1/3$ in the Andrade model, we
obtain an expression in terms of Miller-Ross functions. The asymptotic
behaviours of $G_{\alpha }\left( t\right) $ for $t\rightarrow 0^{+}$ and $t\rightarrow +\infty $ are also derived applying the Tauberian theorem. The
analytical results obtained have been numerically checked by solving the Volterra integral equation satisfied by $G_{\alpha }\left(
t\right) $ by using a successive approximation approach, as well as computing
the inverse Laplace transform of $\tilde{G}_{\alpha
}\left( s\right) $ by using Talbot's method.
\end{abstract}

%%%%%%%%%%%%%%%%%%%%%%%%%%%%%%%%%%%%%%%%%%
%\setcounter{section}{-1} %% Remove this when starting to work on the template.

\section{The Andrade Model in Linear Viscoelasticity}

In the framework of linear viscoelasticity theory, a ``transient'' phase of
deformation occurs right after the elastic response in creep phenomena and
is marked by a strain rate that changes over time \cite{Christensen}. Among
the rheological laws that exhibit a transient phase, the Andrade model has
been effectively used to describe the behavior of various materials. This
model was initially introduced by Andrade in 1910 to describe the elongation
of metallic wires under constant tensile stress \cite{Andrade1910}. Its main
feature is a transient that exhibits a fractional power function time
dependence $\sim$$t^{\alpha}$. In their empirical stress--strain relationship,
Andrade proposed the exponent $\alpha =1/3$ \cite{AndradeTercio}.
Nevertheless, later laboratory investigations have shown that values within
the range of $0<\alpha <1$ are indeed possible for certain \mbox{materials \cite%
{Walterova}}. It is worth noting that during the last dozen years, there has been an increasing interest in the Andrade model \cite{Rosti,Pandey}. Using the modern formulation of the Andrade \mbox{model \cite%
{AndradeModel}}, the creep compliance $J_{\alpha }\left( t\right) $ (i.e., the
strain per unit stress)\ is given by%
\begin{equation}
J_{\alpha }\left( t\right) =J_{U}+\beta t^{\alpha }+\frac{t}{\eta },\quad
t\geq 0,  \label{J(t)_Andrade}
\end{equation}%
where $J_{U}$ is the unrelaxed compliance, $\eta $ is the steady state
Newtonian viscosity, $\beta $ is the magnitude of the inelastic
contribution, and $\alpha $ represents the frequency of the compliance. The
number of free parameters that appear in (\ref{J(t)_Andrade})\ can be
reduced by adopting a useful parametrization given in \cite{Castillo}, whose
validity is discussed in \cite{Walterova}. Essentially, this parametrization
performs the following change of variables:%
\begin{equation}
J_{U}=\frac{1}{\mu },\quad \beta =\frac{1}{\mu \,\tau ^{\alpha }},\quad \eta
=\mu \,\tau ,  \label{change_variables}
\end{equation}%
thus, %MDPI: Please confirm whether all unindented formatting in this article should be preserved as provided. <- CONFIRMED
 we obtain%
\begin{equation}
J_{\alpha }\left( t\right) =\frac{1}{\mu }\left[ 1+\left( \frac{t}{\tau }%
\right) ^{\alpha }+\frac{t}{\tau }\right] ,\quad t\geq 0,  \label{J(t)_Spada}
\end{equation}%
where the parameters $\mu ,\tau >0$ have a clear and important physical
meaning. The aim of this paper is to analytically calculate in closed form
the relaxation modulus $G_{\alpha }\left( t\right) $ (i.e., the stress per
unit strain) from the creep compliance $J_{\alpha }\left( t\right) $ given
in (\ref{J(t)_Spada}) by using the inverse Laplace transform.

It is worth noting that the computation of $G_{\alpha }\left( t\right) $
from $J_{\alpha }\left( t\right) $ is also possible by numerically solving
the following Volterra integral equation of the second kind (\cite%
{Mainardi} \mbox{[Eqn. 2.87]}):
\begin{equation}
G_{\alpha }\left( t\right) =\mu\left[1-\int_{0}^{t}\frac{%
dJ_{\alpha }\left( t^{\prime }\right) }{dt^{\prime }}G_{\alpha }\left(
t-t^{\prime }\right) dt^{\prime }\right],  \label{Volterra_eq}
\end{equation}%
{just} as has been done for other models such as that of Jeffrey-Lomnitz
rheological law \cite{MainardiSpada}. Another numerical approach is to
obtain the relaxation modulus in the Laplace domain
$\tilde{G}_{\alpha }\left( s\right) =\mathcal{L}\left[ G_{\alpha} \left(t\right) ;s\right] $,
and then numerically evaluate the inverse Laplace
transform to obtain $G_{\alpha }\left( t\right) $. However, an analytical
solution is more desirable since it shows the role and weight of the model
parameters explicitly. In addition, analytical approaches are generally more
computationally efficient. In the present approach, we just analytically
calculate the inverse Laplace transform of $\tilde{G}_{\alpha }\left(
s\right) $. This result has been applied by some of the authors in
\cite{Spada}, and here we provide the mathematical details of the calculation.

Next, we derive $\tilde{G}_{\alpha }\left( s\right) $ in the Andrade model.
For this purpose, apply the Laplace transform to (\ref{J(t)_Spada}) in order
to obtain%
\begin{equation}
\mathcal{L}\left[ J_{\alpha }\left( t\right) ;s\right] =\tilde{J}_{\alpha
}\left( s\right) =\frac{1}{\mu \,s^{2}}\left[ s+\Gamma \left( 1+\alpha
\right) \tau ^{-\alpha }s^{1-\alpha }+\frac{1}{\tau }\right] .
\label{L[J(t)]}
\end{equation}

However, since the following relation is satisfied in any linear
viscoelasticity \mbox{rheology (\cite{Mainardi}} [Eqn. 2.8])%
\begin{equation}
\tilde{J}_{\alpha }\left( s\right) \tilde{G}_{\alpha }\left( s\right) =\frac{%
1}{s^{2}},  \label{J(s)*G(s)=1/s^2}
\end{equation}%
% where $\tilde{G}_{\alpha }\left( s\right) =\mathcal{L}\left[ G_{\alpha} \left(t\right) ;s\right] $ denotes the Laplace transform of the relaxation modulus,
{we} conclude%
\begin{equation}
\tilde{G}_{\alpha }\left( s\right) =\frac{\mu \,\tau }{s\,\tau +\Gamma
\left( \alpha +1\right) \left( s\,\tau \right) ^{1-\alpha }+1}.
\label{G_alfa_def}
\end{equation}

As mentioned above, the range of values that are interesting for $\alpha $
are between $0$ and $1$, so we restrict our study to $\alpha \in (0,1)$. It
is worth noting that (\ref{G_alfa_def})\ can also be obtained rewriting the
Volterra integral Equation (\ref{Volterra_eq})\ in terms of the Laplace
convolution \mbox{product \cite{ApelblatLaplace}}, i.e.,
\begin{equation}
G_{\alpha }\left( t\right) =\mu \left( 1-G_{\alpha }\left( t\right) \ast
\frac{dJ_{\alpha }\left( t\right) }{dt}\right) ,
\label{Volterra_convolution}
\end{equation}%
{thus} taking the Laplace transform in (\ref{Volterra_convolution})\ and
solving for $\tilde{G}_{\alpha }\left( s\right) $, we arrive at (\ref%
{G_alfa_def}).

The manuscript is organized as follows. In Section \ref{Section: Miller-Ross}%
, we perform the calculation of $G_{\alpha }\left( t\right) $ for $\alpha
=1/3$ (as first suggested by Andrade) in terms of Miller-Ross functions.
Section \ref{Section:
Rabotnov Functions} generalizes this result for
rational $\alpha $ in terms of a finite sum of Mittag--Leffler functions,
which in turn can be expressed as a linear combination of Rabotnov
functions. It is worth highlighting that although in principle $\alpha$ can be a real number, the value it actually acquires in rheological models is a positive fractional number less than unity. Section \ref{Section: Asymptotic} calculates the asymptotic
behaviour of $G_{\alpha }\left( t\right) $ for $t\rightarrow 0^{+}$ and $%
t\rightarrow +\infty $ by using the Tauberian theorem. Section \ref{Section:
Numerical} shows some numerical verifications on the expressions of $%
G_{\alpha }\left( t\right) $ derived in the previous sections, as well as on
their asymptotic behaviors. Finally, we collect our conclusions in Section %
\ref{Section: Conclusions}.

\section{Laplace Inversion in Terms of Miller-Ross Functions for \boldmath{$ \protect%
\alpha =1/3$} \label{Section: Miller-Ross}}

Consider in (\ref{G_alfa_def}) the following change of variables:
\begin{equation}
b=\frac{1}{\tau },\quad c=\frac{\Gamma \left( \alpha +1\right) }{\tau
^{\alpha }},  \label{parameters_b_c}
\end{equation}%
{to} rewrite $\tilde{G}_{\alpha }\left( s\right) $ as%
\begin{equation}
\tilde{G}_{\alpha }\left( s\right) =\frac{\mu }{b+s+c\,s^{1-\alpha }}.
\label{G_a_Miller}
\end{equation}
In order calculate the inverse Laplace transform of $\tilde{G}_{\alpha
}\left( s\right) $ for the case $\alpha =1/3$, i.e.,%
\begin{equation}
\tilde{G}_{1/3}\left( s\right) =\frac{\mu }{b+s+c\,s^{2/3}},
\label{G_1/3_Miller}
\end{equation}%
{apply} the identity
\begin{equation}
\left( x+y\right) \left( x^{2}-xy+y^{2}\right) =x^{3}+y^{3},
\label{Cubic_identity}
\end{equation}%
{taking} $x=b+s$, and $y=c\,s^{2/3}$, to arrive at
\begin{equation}
\frac{1}{\mu }\tilde{G}_{1/3}\left( s\right) =b^{2}\underset{\tilde{R}\left(
s\right) }{\underbrace{\frac{1}{p\left( s\right) }}}+2b\underset{\tilde{Q}%
\left( s\right) =s\,\tilde{R}\left( s\right) }{\underbrace{\,s\frac{1}{%
p\left( s\right) }}}+\underset{s\,\tilde{Q}\left( s\right) }{\underbrace{%
\frac{s^{2}}{p\left( s\right) }}}-bc\,\underset{\tilde{U}_{2/3}\left(
s\right) }{\underbrace{\frac{s^{2/3}}{p\left( s\right) }}}-c\,\underset{s\,%
\tilde{U}_{2/3}\left( s\right) }{\underbrace{s\frac{s^{2/3}}{p\left(
s\right) }}}+c^{2}\underset{\tilde{U}_{4/3}\left( s\right) }{\underbrace{%
\frac{s^{4/3}}{p\left( s\right) }}},  \label{G_1/3_decomposition}
\end{equation}%
where%
\begin{equation}
p\left( s\right) =s^{3}+\left( 3b+c^{3}\right)
s^{2}+3b^{2}s+b^{3}=\prod_{k=1}^{3}\left( s-s_{k}\right) ,  \label{p(s)_def}
\end{equation}%
{and} $s_{k}$ are the roots of the cubic equation $p\left( s\right) =0$, i.e.,%
\begin{equation}
\begin{array}{l}
%TCIMACRO{\TeXButton{TeX field}{\displaystyle}}%
%BeginExpansion
\displaystyle%
%EndExpansion
s_{1}=-\frac{1}{3}\left( M-N-L\right) , \\
%TCIMACRO{\TeXButton{TeX field}{\displaystyle}}%
%BeginExpansion
\displaystyle%
%EndExpansion
s_{2}=-\frac{1}{3}\left( M+e^{i\,\pi /3}N+e^{-i\,\pi /3}L\right) , \\
%TCIMACRO{\TeXButton{TeX field}{\displaystyle}}%
%BeginExpansion
\displaystyle%
%EndExpansion
s_{3}=-\frac{1}{3}\left( M+e^{-i\,\pi /3}N+e^{i\,\pi /3}L\right) ,%
\end{array}
\label{s_i}
\end{equation}%
{being}%
\begin{equation}
\begin{array}{l}
%TCIMACRO{\TeXButton{TeX field}{\displaystyle}}%
%BeginExpansion
\displaystyle%
%EndExpansion
M=3b+c^{3}, \\
%TCIMACRO{\TeXButton{TeX field}{\displaystyle}}%
%BeginExpansion
\displaystyle%
%EndExpansion
L=\sqrt[3]{\frac{3}{2}\sqrt{3b^{3}c^{6}\left( 27b+4c^{3}\right) }-\frac{27}{2%
}b^{2}c^{3}-9\,bc^{6}-c^{9}}, \\
%TCIMACRO{\TeXButton{TeX field}{\displaystyle}}%
%BeginExpansion
\displaystyle%
%EndExpansion
N=\frac{\left( 6b+c^{3}\right) c^{3}}{L}.%
\end{array}
\label{MLN_def}
\end{equation}
First, rewrite $\tilde{R}\left( s\right) $ as
\begin{equation}
\tilde{R}\left( s\right) =\frac{1}{p\left( s\right) }=\sum_{k=1}^{3}\frac{%
\alpha _{k}}{s-s_{k}},\quad \alpha _{k}=\prod_{j\neq k}^{3}\frac{1}{%
s_{k}-s_{j}}.  \label{R(s)_def}
\end{equation}
Note that a simple algebraic calculation shows that%
\begin{equation}
\begin{array}{l}
%TCIMACRO{\TeXButton{TeX field}{\displaystyle}}%
%BeginExpansion
\displaystyle%
%EndExpansion
\sum_{k=1}^{3}\alpha _{k}=0, \\
%TCIMACRO{\TeXButton{TeX field}{\displaystyle}}%
%BeginExpansion
\displaystyle%
%EndExpansion
\sum_{k=1}^{3}\alpha _{k}\,s_{k}=0.%
\end{array}
\label{Sum_alfa=0}
\end{equation}
Now, define the following functions%
\begin{equation}
R\left( t\right) =\mathcal{L}^{-1}\left[ \tilde{R}\left( s\right) ;t\right] ,
\label{r(t)_def}
\end{equation}%
{thus,}%
\begin{equation}
Q\left( t\right) =\mathcal{L}^{-1}\left[ \tilde{Q}\left( s\right) ;t\right] =%
\mathcal{L}^{-1}\left[ s\,\tilde{R}\left( s\right) ;t\right] =R^{\prime
}\left( t\right) +R\left( 0\right) \,\delta \left( t\right) ,
\label{Q(t)_def}
\end{equation}%
{and}%
\begin{equation}
\mathcal{L}^{-1}\left[ s\,\tilde{Q}\left( s\right) ;t\right] =Q^{\prime
}\left( t\right) +Q\left( 0\right) \delta \left( t\right) =R^{\prime \prime
}\left( t\right) +R^{\prime }\left( 0\right) \,\delta ^{\prime }\left(
t\right) .  \label{L-1[s*Q(s)]}
\end{equation}
Furthermore,%
\begin{equation}
U_{\nu }\left( t\right) =\mathcal{L}^{-1}\left[ \tilde{U}_{\nu }\left(
s\right) ;t\right] =\mathcal{L}^{-1}\left[ \frac{s^{\nu }}{p\left( s\right) }%
;t\right] ,  \label{u_nu(t)_def}
\end{equation}%
{thus}%
\begin{equation}
\mathcal{L}^{-1}\left[ s\,\tilde{U}_{\nu }\left( s\right) ;t\right] =U_{\nu
}^{\prime }\left( t\right) +U_{\nu }\left( 0\right) \,\delta \left( t\right)
.  \label{L-1[s*U]}
\end{equation}
Therefore, the inverse Laplace transform of $\tilde{G}_{1/3}\left( s\right) $
is given by
\begin{eqnarray}
G_{1/3}\left( t\right) &=&\mathcal{L}^{-1}\left[ \tilde{G}_{1/3}\left(
s\right) ;t\right]  \label{g_1/3_decomposition} \\
&=&\mu \left\{ b^{2}\,R\left( t\right) +2b\,\left[ R^{\prime }\left(
t\right) +R\left( 0\right) \,\delta \left( t\right) \right] +R^{\prime
\prime }\left( t\right) +R^{\prime }\left( 0\right) \,\delta ^{\prime
}\left( t\right) \right.  \notag \\
&&+\left. bc\,U_{2/3}\left( t\right) -c\left[ U_{2/3}^{\prime }\left(
t\right) +U_{2/3}\left( 0\right) \,\delta \left( t\right) \right]
+c^{2}\,U_{4/3}\left( t\right) \right\} .  \notag
\end{eqnarray}

\subsection{Calculation of $R \left( t \right) $}

According to (\ref{R(s)_def}) and (\ref{r(t)_def}), we have%
\begin{equation}
R\left( t\right) =\mathcal{L}^{-1}\left[ \tilde{R}\left( s\right) ;t\right]
=\sum_{k=1}^{3}\alpha _{k}\mathcal{L}^{-1}\left[ \frac{1}{s-s_{k}};t\right]
=\sum_{k=1}^{3}\alpha _{k}\exp \left( s_{k}t\right) ,  \label{r(t)_resultado}
\end{equation}%
{so,} that%
\begin{equation}
\begin{array}{c}
%TCIMACRO{\TeXButton{TeX field}{\displaystyle}}%
%BeginExpansion
\displaystyle%
%EndExpansion
R^{\prime }\left( t\right) =\sum_{k=1}^{3}\alpha _{k}\,s_{k}\exp \left(
s_{k}t\right) , \\
%TCIMACRO{\TeXButton{TeX field}{\displaystyle}}%
%BeginExpansion
\displaystyle%
%EndExpansion
R^{\prime \prime }\left( t\right) =\sum_{k=1}^{3}\alpha _{k}\,s_{k}^{2}\exp
\left( s_{k}t\right) .%
\end{array}
\label{r''(t)_resultado}
\end{equation}
Therefore, from (\ref{Sum_alfa=0}), we obtain%
\begin{equation}
\begin{array}{l}
%TCIMACRO{\TeXButton{TeX field}{\displaystyle}}%
%BeginExpansion
\displaystyle%
%EndExpansion
R\left( 0\right) =\sum_{k=1}^{3}\alpha _{k}=0, \\
%TCIMACRO{\TeXButton{TeX field}{\displaystyle}}%
%BeginExpansion
\displaystyle%
%EndExpansion
R^{\prime }\left( 0\right) =\sum_{k=1}^{3}\alpha _{k}\,s_{k}=0.%
\end{array}
\label{r(0)=0}
\end{equation}

\subsection{Calculation of $U_{ \protect \nu } \left( t \right) $}

According to (\ref{u_nu(t)_def}), (\ref{R(s)_def}), and the inverse Laplace
transform (\cite{Prudnikov5} [Eqn. 2.1.2(9)])
\begin{equation}
\mathcal{L}^{-1}\left[ \frac{s^{\nu }}{s-a};t\right] =a^{\nu }e^{at}P\left(
-\nu ,at\right) ,  \label{L-1[s^nu/(s-a)]}
\end{equation}%
where
\begin{equation}
P\left( \nu ,x\right) =\frac{\gamma \left( \nu ,x\right) }{\Gamma \left( \nu
\right) }=\frac{1}{\Gamma \left( \nu \right) }\int_{0}^{x}z^{\nu -1}e^{-z}dz,
\label{P(nu,x)_def}
\end{equation}%
{denotes} the \textit{normalized lower incomplete gamma function} %MDPI: Is the italics necessary? BETTER IN ITALICS
 (\cite{NIST} [Eqn.
8.2.4]), we have%
\begin{eqnarray}
U_{\nu }\left( t\right) &=&\mathcal{L}^{-1}\left[ \frac{s^{\nu }}{p\left(
s\right) };t\right]  \notag \\
&=&\sum_{k=1}^{3}\alpha _{k}\,\mathcal{L}^{-1}\left[ \frac{s^{\nu }}{s-s_{k}}%
;t\right]  \notag \\
&=&\sum_{k=1}^{3}\alpha _{k}\,s_{k}^{\nu }\exp \left( s_{k}t\right)
\,P\left( -\nu ,s_{k}t\right) .  \label{u_nu_resultado}
\end{eqnarray}
Therefore, apply the property (\cite{NIST} [Eqn. 8.4.2])%
\begin{equation}
\lim_{x\rightarrow 0}x^{\nu }P\left( -\nu ,x\right) =\frac{1}{\Gamma \left(
1-\nu \right) },  \label{lim_P(nu,x)}
\end{equation}%
{to} calculate, according to (\ref{Sum_alfa=0}), that
\begin{eqnarray}
U_{\nu }\left( 0\right) &=&\lim_{t\rightarrow 0}t^{-\nu
}\sum_{k=1}^{3}\alpha _{k}\,\left( s_{k}t\right) ^{\nu }P\left( -\nu
,s_{k}t\right)  \notag \\
&=&\lim_{t\rightarrow 0}\frac{t^{-\nu }}{\Gamma \left( 1-\nu \right) }%
\sum_{k=1}^{3}\alpha _{k}=0.  \label{u_nu(0)=0}
\end{eqnarray}
Furthermore, from the derivative formula
\begin{equation}
\frac{d}{dt}\left[ e^{at}P\left( -\nu ,at\right) \right] =a\left[
e^{at}P\left( -\nu ,at\right) +\frac{\left( at\right) ^{-\nu -1}}{\Gamma
\left( -\nu \right) }\right] ,  \label{Derivative_formula}
\end{equation}%
{we} calculate, according also to (\ref{Sum_alfa=0}), that%
\begin{eqnarray}
U_{\nu }^{\prime }\left( t\right) &=&\sum_{k=1}^{3}\alpha _{k}\,s_{k}^{\nu
+1}\left[ \exp \left( s_{k}t\right) \,P\left( -\nu ,s_{k}t\right) +\frac{%
\left( s_{k}t\right) ^{-\nu -1}}{\Gamma \left( -\nu \right) }\right]  \notag
\\
&=&\sum_{k=1}^{3}\alpha _{k}\,s_{k}^{\nu +1}\exp \left( s_{k}t\right)
\,P\left( -\nu ,s_{k}t\right) +\frac{t^{-\nu -1}}{\Gamma \left( -\nu \right)
}\sum_{k=1}^{3}\alpha _{k}  \notag \\
&=&\sum_{k=1}^{3}\alpha _{k}\,s_{k}^{\nu +1}\exp \left( s_{k}t\right)
\,P\left( -\nu ,s_{k}t\right) .  \label{u'_nu(t)_resultado}
\end{eqnarray}

\subsection{Calculation of $G_{1/3} \left( t \right) $}

Insert (\ref{r(t)_resultado})--(\ref{u'_nu(t)_resultado}) into (\ref%
{g_1/3_decomposition}) to arrive at the following result, after
simplification,
\begin{eqnarray}
&&G_{1/3}\left( t\right) =\mu \sum_{k=1}^{3}\alpha _{k}\exp \left(
s_{k}t\right)  \label{g_1/3(t)_resultado} \\
&&\quad \left\{ \left( b+s_{k}\right) \left[ b+s_{k}-c\,s_{k}^{2/3}P\left( -%
\frac{2}{3},s_{k}t\right) \right] +c^{2}s_{k}^{4/3}P\left( -\frac{4}{3}%
,s_{k}t\right) \right\} ,\quad t\geq 0.  \notag
\end{eqnarray}

It is worth noting that we can rewrite (\ref{g_1/3(t)_resultado})\ in terms
of the Miller-Ross functions (\cite{Mainardi} [Eqn. E.37]), defined as%
\begin{equation}
\mathrm{E}_{t}\left( \nu ,a\right) =\frac{a^{-\nu }e^{at}}{\Gamma \left( \nu
\right) }\,\gamma \left( \nu ,at\right) =a^{-\nu }e^{at}\,P\left( \nu
,at\right) ,  \label{Miller-Ross_def}
\end{equation}%
thus, after simplification, we arrive at the following result:%
\begin{equation}
G_{1/3}^{\mathrm{MR}}\left( t\right) =\mu \sum_{k=1}^{3}\alpha _{k}\left\{
\left( b+s_{k}\right) \left[ \left( b+s_{k}\right) e^{s_{k}t}-c\,\mathrm{E}%
_{t}\left( -\frac{2}{3},s_{k}\right) \right] +c^{2}\,\mathrm{E}_{t}\left( -%
\frac{4}{3},s_{k}t\right) \right\} ,  \label{G_tercio_MR}
\end{equation}%
where the superscript $\mathrm{MR}$ takes into account that $G_{1/3}\left(
t\right) $ is given in terms of Miller-Ross functions. Furthermore, according to (%
\ref{parameters_b_c}), the parameters $b=1/\tau $ and $c=\Gamma \left( \frac{%
4}{3}\right) \tau ^{-1/3}$. Moreover, $s_{k}$ and $\alpha_{k}$ are given in (\ref%
{s_i})--(\ref{R(s)_def}), respectively.

\section{Laplace Inversion in Terms of Rabotnov Functions} \label{Section:
Rabotnov Functions}

Consider (\ref{G_alfa_def}) for the case $\alpha =\frac{m}{n}\in
%TCIMACRO{\U{211a} }%
%BeginExpansion
\mathbb{Q}
%EndExpansion
$,
\begin{equation}
0<m<n,\quad n,m\in
%TCIMACRO{\U{2115} }%
%BeginExpansion
\mathbb{N}
%EndExpansion
,  \label{0<m<n}
\end{equation}%
and perform the change of variables $r=\left( s\tau \right) ^{1/n}$,%
\begin{equation}
\tilde{G}_{m/n}\left( r\right) =\frac{\mu \,\tau }{\underset{p_{n,m}\left(
r\right) }{\underbrace{r^{n}+\Gamma \left( 1+\frac{m}{n}\right) r^{n-m}+1}}},
\label{G_m/n_def}
\end{equation}%
where $p_{n,m}\left( r\right) $ is a polynomial of $n$-th order. If $%
p_{n,m}\left( r\right) $ has non-repeated roots $r_{k}$, $\left( k=1,\ldots
,n\right) $, then, according to (\cite{Atlas} [Eqn. 17:13:10]), we have%
\begin{equation}
\frac{1}{p_{n,m}\left( r\right) }=\sum_{k=1}^{n}\frac{1}{p_{n,m}^{\prime
}\left( r_{k}\right) \left( r-r_{k}\right) },  \label{1/p_result}
\end{equation}%
{thus}%
\begin{equation}
\tilde{G}_{m/n}\left( r\right) =\mu \,\tau \sum_{k=1}^{n}\frac{1}{%
p_{n,m}^{\prime }\left( r_{k}\right) \left( r-r_{k}\right) },
\label{G_m/n_p}
\end{equation}%
{and}%
\begin{equation}
\tilde{G}_{m/n}\left( s\right) =\mu \,\tau \sum_{k=1}^{n}\frac{1}{%
p_{n,m}^{\prime }\left( r_{k}\right) \left( s^{1/n}\tau ^{1/n}-r_{k}\right) }%
.  \label{G_m/n_s}
\end{equation}
Define,
\begin{eqnarray}
f_{m/n}\left( t\right) &=&\mathcal{L}^{-1}\left[ \tilde{G}_{m/n}\left(
s/\tau \right) ;t\right]  \notag \\
&=&\mu \,\tau \sum_{k=1}^{n}\frac{1}{p_{n,m}^{\prime }\left( r_{k}\right) }%
\mathcal{L}^{-1}\left[ \frac{1}{s^{1/n}-r_{k}}\right] ,  \label{f_1/n_def}
\end{eqnarray}%
{and} apply (\cite{Atlas} [Eqn. 45:14:4])%
\begin{equation}
\mathcal{L}^{-1}\left[ \frac{s^{\mu -\nu }}{s^{\mu }-a};t\right] =t^{\nu
-1}\,\mathrm{E}_{\mu ,\nu }\left( a\,t^{\mu }\right) ,  \label{L-1->ML}
\end{equation}%
where $\mathrm{E}_{\alpha ,\beta }\left( z\right) $ denotes the
two-parameter Mittag--Leffler function (\cite{Atlas} [Eqn. 45:14:2]),
\begin{equation}
\mathrm{E}_{\alpha ,\beta }\left( z\right) =\sum_{k=0}^{\infty }\frac{z^{k}}{%
\Gamma \left( \beta +\alpha \,k\right) }.  \label{ML_def}
\end{equation}%
Therefore, for $\mu =\nu =\frac{1}{n}$, we have%
\begin{equation}
\mathcal{L}^{-1}\left[ \frac{1}{s^{1/n}-a};t\right] =t^{1/n-1}\mathrm{E}_{%
\frac{1}{n},\frac{1}{n}}\left( a\,t^{1/n}\right) .  \label{L-1_ML}
\end{equation}
Insert (\ref{L-1_ML})\ in (\ref{f_1/n_def})\ to arrive at%
\begin{equation}
f_{m/n}\left( t\right) =\mu \,\tau \,t^{1/n-1}\sum_{k=1}^{n}\frac{\mathrm{E}%
_{\frac{1}{n},\frac{1}{n}}\left( r_{k}\,t^{1/n}\right) }{p_{n,m}^{\prime
}\left( r_{k}\right) }.  \label{f_m/n(t)_def}
\end{equation}
Finally, apply the property (\cite{Prudnikov5} [Eqn. 1.1.1(3)]),
\begin{equation}
\mathcal{L}^{-1}\left[ \tilde{F}\left( s\right) ;t\right] =F\left( t\right)
\ \leftrightarrow \ \mathcal{L}^{-1}\left[ \tilde{F}\left( a\,s\right) ;t%
\right] =\frac{1}{a}F\left( \frac{t}{a}\right) ,  \label{L-1_property}
\end{equation}%
{to} obtain%
\begin{equation}
G_{m/n}\left( t\right) =\mathcal{L}^{-1}\left[ \tilde{G}_{m/n}\left(
s\right) ;t\right] =\frac{1}{\tau }\,f_{m/n}\left( \frac{t}{\tau }\right) ,
\label{G_m/n->f_m/n}
\end{equation}%
i.e.,%
\begin{equation}
G_{m/n}\left( t\right) =\mu \,\left( \frac{t}{\tau }\right)
^{1/n-1}\sum_{k=1}^{n}\frac{\mathrm{E}_{\frac{1}{n},\frac{1}{n}}\left(
r_{k}\,\left( t/\tau \right) ^{1/n}\right) }{p_{n,m}^{\prime }\left(
r_{k}\right) },\quad t\geq 0,  \label{G_m/n(t)_resultado}
\end{equation}%
where remember that $r_{k}$ are the $n$ non-repeated roots of the
polynomial:
\begin{equation}
p_{n,m}\left( r\right) =r^{n}+\Gamma \left( 1+\frac{m}{n}\right) r^{n-m}+1.
\label{p_n,m(r)_def}
\end{equation}
Note that the solution is a linear combination of Rabotnov functions (\cite{Mainardi} [%
Eqn. E.46]):\
\begin{equation}
\mathrm{R}_{\nu }\left( \mu ,t\right) =t^{\nu }\,\mathrm{E}_{\nu +1,\nu
+1}\left( \mu \,t^{\nu +1}\right) ,  \label{Rabotnov_def}
\end{equation}%
{thereby},%
\begin{equation}
G_{m/n}\left( t\right) =\mu \sum_{k=1}^{n}\frac{\mathrm{R}_{1/n-1}\left(
r_{k},t/\tau \right) }{p_{n,m}^{\prime }\left( r_{k}\right) }.
\label{G_m/n_Rabotnov}
\end{equation}
For the particular case $\alpha =m/n=1/3$, we have%
\begin{equation}
G_{1/3}^{\mathrm{R}}\left( t\right) =\mu \sum_{k=1}^{3}\frac{\mathrm{R}%
_{-2/3}\left( r_{k},t/\tau \right) }{2\,\Gamma \left( \frac{4}{3}\right)
r_{k}+3\,r_{k}^{2}},\quad t\geq 0,  \label{GR_1/3(t)_resultado}
\end{equation}%
where the superscript $\mathrm{R}$ takes into account that $G_{1/3}\left(
t\right) $ is given in terms of Rabotnov functions, and $r_{k}$ are the
three distinct roots of the cubic equation,
\begin{equation}
r^{3}+\Gamma \left( \frac{4}{3}\right) r^{2}+1=0,  \label{cubic_eq}
\end{equation}%
i.e.,%
\begin{equation}
\begin{array}{l}
r_{1}\approx -1.40184, \\
r_{2}\approx 0.254432-0.805364\,i, \\
r_{3}\approx 0.254432+0.805364\,i.%
\end{array}
\label{r_i}
\end{equation}

\section{Asymptotic Behaviour via Tauberian Theorem} \label{Section:
Asymptotic}

Next, we will obtain the asymptotic behaviour of the relaxation modulus $%
G_{\alpha }\left( t\right) $ as $t\rightarrow 0^{+}$ and as $t\rightarrow
+\infty $ from its Laplace transform $\tilde{G}_{\alpha }\left( s\right) $
by using the following version of the Tauberian theorem, (for other version of the Tauberian theorem, see \cite{Sandev}).

\begin{theorem}
Consider that the Laplace transform of a function $f\left( t\right) $ is
given by $\tilde{f}\left( s\right) =\mathcal{L}\left[ f\left( t\right) ;s%
\right] $. The asymptotic behaviour of $\tilde{f}\left( s\right) $ as $%
s\rightarrow +\infty $ is given by \
\begin{equation}
\tilde{f}\left( s\right) \approx \mathcal{L}\left[ g\left( t\right) ;s\right]
,\quad s\rightarrow +\infty ,  \label{Theorem_s->inf}
\end{equation}%
where $g\left( t\right) $ is the asymptotic behaviour of $f\left( t\right) $
as $t\rightarrow 0^{+}$. Furthermore, the asymptotic behaviour of $\tilde{f}\left(
s\right) $ as $s\rightarrow 0^{+}$ is given by \
\begin{equation}
\tilde{f}\left( s\right) \approx \mathcal{L}\left[ h\left( t\right) ;s\right]
,\quad s\rightarrow 0^{+},  \label{Theorem_s->0}
\end{equation}%
where $h\left( t\right) $ is the asymptotic behaviour of $f\left( t\right) $
as $t\rightarrow +\infty $.
\end{theorem}

\begin{proof}
On the one hand, consider that the asymptotic behaviour of $f\left( t\right)
$ as $t\rightarrow 0^{+}$ is \mbox{given by}
\begin{equation}
f\left( t\right) \approx g\left( t\right)
=\sum_{k=1}^{N}a_{k}\,t^{b_{k}},\quad t\rightarrow 0^{+},  \label{f(t)_t->0}
\end{equation}%
{with} $N=0,1,2,\ldots $, and $0\leq b_{1}<b_{2}<\cdots <b_{N}$. Apply the
Laplace transform to (\ref{f(t)_t->0}) in order to obtain%
\begin{equation*}
\tilde{f}\left( s\right) =\mathcal{L}\left[ f\left( t\right) ;s\right]
\approx \mathcal{L}\left[ g\left( t\right) ;s\right] =\sum_{k=1}^{N}a_{k}\,%
\mathcal{L}\left[ t^{b_{k}};s\right] =\sum_{k=1}^{N}\frac{a_{k}\,\Gamma
\left( b_{k}+1\right) }{s^{b_{k}+1}}.
\end{equation*}%
Therefore, we obtain the asymptotic behaviour of $\tilde{f}\left( s\right) $
as $s\rightarrow +\infty $,
\begin{equation*}
\tilde{f}\left( s\right) \approx \sum_{k=0}^{N}\frac{a_{k}\,\Gamma \left(
b_{k}+1\right) }{s^{b_{k}+1}},\quad s\rightarrow +\infty ,
\end{equation*}%
{as} we wanted to prove.

On the other hand, consider that the asymptotic behaviour of $\tilde{f}%
\left( s\right) $ as $s\rightarrow 0^{+}$ is given by%
\begin{equation}
\tilde{f}\left( s\right) \approx \sum_{k=1}^{N}c_{k}\,s^{d_{k}},\quad
s\rightarrow 0^{+},  \label{F(s)_s->inf}
\end{equation}%
{with} $N=0,1,2,\ldots $, and $0\leq d_{1}<d_{2}<\cdots <d_{N}$. Apply the
inverse Laplace transform\linebreak  to (\ref{F(s)_s->inf}) in order to obtain
\begin{equation*}
f\left( t\right) =\mathcal{L}^{-1}\left[ \tilde{f}\left( s\right) ;t\right]
=\sum_{k=1}^{N}c_{k}\,\mathcal{L}^{-1}\left[ s^{d_{k}};t\right]
=\sum_{k=1}^{N}\frac{c_{k}}{\Gamma \left( -d_{k}\right) \,t^{d_{k}+1}}\,,
\end{equation*}%
Therefore, we obtain the asymptotic behaviour of $f\left( t\right) $ as $%
t\rightarrow +\infty $,%
\begin{equation*}
f\left( t\right) \approx h\left( t\right) =\sum_{k=1}^{N}\frac{c_{k}}{\Gamma
\left( -d_{k}\right) \,t^{d_{k}+1}},\quad t\rightarrow +\infty ,
\end{equation*}%
{as} we wanted to prove.
\end{proof}

\subsection{Asymptotic Behaviour for $t \rightarrow + \infty $}

We know that the Laplace transform of the relaxation modulus in the Andrade\linebreak
model is (\ref{G_alfa_def})
\begin{equation}
\tilde{G}_{\alpha }\left( s\right) =\frac{\mu \,\tau }{s\,\tau +\Gamma
\left( \alpha +1\right) \left( s\,\tau \right) ^{1-\alpha }+1},\quad \alpha
\in \left( 0,1\right) .  \label{G_a(s)}
\end{equation}
Since
\begin{equation}
\frac{1}{1-x}=1+x+x^{2}+\cdots ,\quad \left\vert x\right\vert <1,
\label{geom_series}
\end{equation}%
{we} have that%
\begin{equation}
\frac{1}{1+x}\approx 1-x,\quad x\rightarrow 0^{+},  \label{Asympt_x->0}
\end{equation}%
{thus}, taking $x=s\,\tau +\Gamma \left( \alpha +1\right) \left( s\,\tau
\right) ^{1-\alpha }$, we get%
\begin{equation}
\tilde{G}_{\alpha }\left( s\right) \approx \mu \,\tau \left[ 1-s\,\tau
-\Gamma \left( \alpha +1\right) \left( s\,\tau \right) ^{1-\alpha }\right]
,\quad s\rightarrow 0^{+}.  \label{G(s)_s->0_(a)}
\end{equation}
Since $\alpha \in \left( 0,1\right) $,
\begin{equation}
\tilde{G}_{\alpha }\left( s\right) \approx \mu \,\tau \left[ 1-\Gamma \left(
\alpha +1\right) \,\left( s\,\tau \right) ^{1-\alpha }\right] ,\quad
s\rightarrow 0^{+}.  \label{G(s)_s->0_(b)}
\end{equation}
According to (\ref{Theorem_s->0}), the asymptotic behaviour of $G_{\alpha
}\left( t\right) $ as $t\rightarrow +\infty $ is calculated as%
\begin{eqnarray}
G_{\alpha }\left( t\right) &\approx &\mu \,\tau \mathcal{L}^{-1}\left[
1-\Gamma \left( \alpha +1\right) \,\left( s\,\tau \right) ^{1-\alpha };t%
\right] \\ \notag
&=&\mu \,\tau \,\delta \left( t\right) -\mu \,\frac{\Gamma \left( \alpha
+1\right) }{\Gamma \left( \alpha -1\right) }\left( \frac{t}{\tau }\right)
^{\alpha -2}. \label{G_a(t)_asymp_(1)}
\end{eqnarray}
Applying the factorial property of the Gamma function, $\Gamma \left(
z+1\right) =z\,\Gamma \left( z\right) $, and taking into account that $%
\delta \left( t\right) =0$ as $t\rightarrow +\infty $, we conclude that%
\begin{equation}
G_{\alpha }\left( t\right) \approx \mu \,\alpha \left( 1-\alpha \right)
\left( \frac{t}{\tau }\right) ^{\alpha -2},\quad t\rightarrow +\infty .
\label{G_a(t)_asymp}
\end{equation}
Taking more terms in the expansion of $\tilde{G}_{\alpha }\left( s\right) $
as $s\rightarrow 0^{+}$, we can calculate more terms of $G_{\alpha }\left(
t\right) $ as $t\rightarrow +\infty $ by using (\ref{Theorem_s->0}).
Thereby, we obtain%
\begin{equation}
G_{\alpha }\left( t\right) \approx \mu \,\alpha \left[ \left( 1-\alpha
\right) \left( \frac{t}{\tau }\right) ^{\alpha -2}-2\,\frac{\Gamma
^{2}\left( 1+\alpha \right) }{\Gamma \left( 2\alpha -1\right) }\left( \frac{t%
}{\tau }\right) ^{2\alpha -3}+\cdots \right] ,\quad t\rightarrow +\infty .
\label{G_a(t)_t->inf_general}
\end{equation}
For the particular case $\alpha =\frac{1}{3}$ in (\ref{G_a(t)_asymp}), we
have
\begin{equation}
G_{1/3}\left( t\right) \approx \frac{2}{9}\,\mu \left( \frac{t}{\tau }%
\right) ^{-5/3},\quad t\rightarrow +\infty .  \label{Asymp_G_t->inf}
\end{equation}
As a consistency test, we can obtain (\ref{Asymp_G_t->inf})\ from the the
expression given in (\ref{G_m/n(t)_resultado}) for $G_{m/n}\left( t\right) $
with $m=1$ and $n=3$, and the asymptotic formula (\cite{Ederly} [Eqn. 18.1(22)])%
,%
\begin{equation}
\mathrm{E}_{\alpha ,\beta }\left( z\right) \approx -\sum_{k=1}^{N}\frac{%
z^{-k}}{\Gamma \left( \beta -\alpha k\right) }+O\left( \left\vert
z\right\vert ^{-N}\right) ,\quad z\rightarrow \infty ,\quad \left\vert \arg
\left( -z\right) \right\vert <\left( 1-\frac{\alpha }{2}\right) \pi .
\label{ML_t->inf}
\end{equation}

\subsection{Asymptotic Behaviour for $t \rightarrow 0^{+}$}

Rewrite the Laplace transform of the relaxation modulus (\ref{G_a(s)}), as
follows:%
\begin{equation}
\tilde{G}_{\alpha }\left( s\right) =\frac{\mu }{s}\left[ \frac{1}{1+\Gamma
\left( \alpha +1\right) \left( s\,\tau \right) ^{-\alpha }+\left( s\,\tau
\right) ^{-1}}\right] .  \label{G_a(t)_2}
\end{equation}
Note that, for $a\in \left( 0,1\right) $ (thus $1-a\in \left( 0,1\right) $),
and $A\in
%TCIMACRO{\U{211d} }%
%BeginExpansion
\mathbb{R}
%EndExpansion
$, we have that%
\begin{equation}
\lim_{y\rightarrow +\infty }\frac{1+A\,y^{-a}+y^{-1}}{1+A\,y^{-a}}%
=\lim_{y\rightarrow +\infty }\frac{y+A\,y^{1-a}+1}{y+A\,y\,^{1-a}}=1,
\label{limit_formula}
\end{equation}%
{thus}, taking $A=\Gamma \left( \alpha +1\right) $, $a=\alpha \in \left(
0,1\right) $, and $y=s\tau \rightarrow +\infty $, (i.e., $s\rightarrow
+\infty $, since $\tau >0$), we get
\begin{equation}
1+\Gamma \left( \alpha +1\right) \left( s\,\tau \right) ^{-\alpha }+\left(
s\,\tau \right) ^{-1}\approx 1+\Gamma \left( \alpha +1\right) \left( s\,\tau
\right) ^{-\alpha },\quad s\rightarrow +\infty .  \label{Asymp_formula}
\end{equation}
Apply (\ref{Asymp_formula})\ to (\ref{G_a(t)_2}), in order to obtain%
\begin{equation}
\tilde{G}_{\alpha }\left( s\right) \approx \frac{\mu }{s}\left[ \frac{1}{%
1+\Gamma \left( \alpha +1\right) \left( s\,\tau \right) ^{-\alpha }}\right]
,\quad s\rightarrow +\infty .  \label{G(s)_s->inf_(b)}
\end{equation}
Now, perform the change of variables $x=1/z$ in (\ref{Asympt_x->0}),
\begin{equation}
\frac{1}{1+\frac{1}{z}}\approx 1-\frac{1}{z},\quad z\rightarrow +\infty ,
\label{geom_1/z}
\end{equation}%
{and} take $z=\frac{\left( s\,\tau \right) ^{\alpha }}{\Gamma \left( \alpha
+1\right) }\rightarrow +\infty $ (i.e., $s\rightarrow +\infty $), to arrive at%
\begin{equation}
\tilde{G}_{\alpha }\left( s\right) \approx \mu \left[ \frac{1}{s}-\frac{%
\Gamma \left( \alpha +1\right) }{\tau ^{\alpha }s\,^{\alpha +1}}\right]
,\quad s\rightarrow +\infty .  \label{G(s)_s->inf_(c)}
\end{equation}
According to (\ref{Theorem_s->inf}), the asymptotic behaviour of $G_{\alpha
}\left( t\right) $ as $t\rightarrow 0^{+}$ is calculated as%
\begin{equation}
G_{\alpha }\left( t\right) \approx \mu \,\mathcal{L}^{-1}\left[ \frac{1}{s}-%
\frac{\Gamma \left( \alpha +1\right) }{\tau ^{\alpha }s\,^{\alpha +1}};t%
\right] , \label{G_a(t)_inv}
\end{equation}%
{i.e.,}
\begin{equation}
G_{\alpha }\left( t\right) \approx \mu \left[ 1-\left( \frac{t}{\tau }%
\right) ^{\alpha }\right] ,\quad t\rightarrow 0^{+}.  \label{G_a(t)_t->0}
\end{equation}
Again, taking more terms in the expansion of $\tilde{G}_{\alpha }\left(
s\right) $ as $s\rightarrow +\infty $, we can calculate more terms of $%
G_{\alpha }\left( t\right) $ as $t\rightarrow 0^{+}$ by using (\ref%
{Theorem_s->inf}). Thereby, we obtain%
\begin{equation}
G_{\alpha }\left( t\right) \approx \mu \left[ 1-\left( \frac{t}{\tau }%
\right) ^{\alpha }+\,\frac{\Gamma ^{2}\left( 1+\alpha \right) }{\Gamma
\left( 2\alpha -1\right) }\left( \frac{t}{\tau }\right) ^{2\alpha }+\cdots %
\right] ,\quad t\rightarrow 0^{+}.  \label{G_a(t)_t->0_general}
\end{equation}
Note that the particular case $\alpha =\frac{1}{3}$ in (\ref{G_a(t)_t->0})
yields
\begin{equation}
G_{1/3}\left( t\right) \approx \mu \left[ 1-\left( \frac{t}{\tau }\right)
^{1/3}\right] ,\quad t\rightarrow 0^{+}.  \label{Asymp_G_t->0}
\end{equation}

As a consistency test, we can obtain (\ref{Asymp_G_t->0})\ from the the
expression given in (\ref{G_m/n(t)_resultado}) for $G_{m/n}\left( t\right) $
with $m=1$ and $n=3$, and the definition of the Mittag--Leffler Function (\ref%
{ML_def}). Furthermore, the asymptotic formula given in (\ref{G_a(t)_t->0})\ can be
obtained from the Volterra integral \mbox{Equation (\ref{Volterra_eq})}. Indeed,
taking into account (\ref{change_variables}) and (\ref{J(t)_Spada}) and
performing the change of variables $x=t-t^{\prime }$, this integral equation
reads as
\begin{equation}
G_{\alpha }\left( t\right) =\mu -\frac{1}{\tau }\int_{0}^{t}\left[ 1+\alpha
\left( \frac{t-x}{\tau }\right) ^{\alpha -1}\right] G_{\alpha }\left(
x\right) \,dx,  \label{Volterra_Ga(t)}
\end{equation}%
{thus,}
\begin{equation}
\lim_{t\rightarrow 0^{+}}G_{\alpha }\left( t\right) =\mu .  \label{G(0) = mu}
\end{equation}
According to (\ref{G(0) = mu}), we can take the approximation $G_{\alpha
}\left( x\right) \approx \mu $ as $t\rightarrow 0^{+}$ in (\ref%
{Volterra_Ga(t)}), thus
\begin{eqnarray}
G_{\alpha }\left( t\right) &\approx &\mu -\frac{\mu }{\tau }\int_{0}^{t}%
\left[ 1+\alpha \left( \frac{t-x}{\tau }\right) ^{\alpha -1}\right] \,dx
\label{G_Volterra_int} \\
&=&\mu \left[ 1-\frac{t}{\tau }-\left( \frac{t}{\tau }\right) ^{\alpha }%
\right]  \notag
\end{eqnarray}
Recalling that $\alpha \in \left( 0,1\right) $, we recover (\ref{G_a(t)_t->0}%
), i.e., \
\begin{equation}
G_{\alpha }\left( t\right) \approx \mu \left[ 1-\left( \frac{t}{\tau }%
\right) ^{\alpha }\right] ,\quad t\rightarrow 0^{+}.
\label{G(t)_t->0_Volterra}
\end{equation}

Figure \ref{Figure: G_tercio_plot}\ presents the graph of $G_{1/3}\left(
t\right) $ for $\mu =1$ and different values of $\tau $. Figure \ref{Figure: G tercio asymp} shows the asymptotic behaviours given in (\ref{Asymp_G_t->0})\
and (\ref{Asymp_G_t->inf}) for $G_{1/3}\left( t\right) $ with $\mu =1$ and $%
\tau =\frac{1}{2}$.

\begin{figure}[htbp]
\includegraphics[width=10cm]{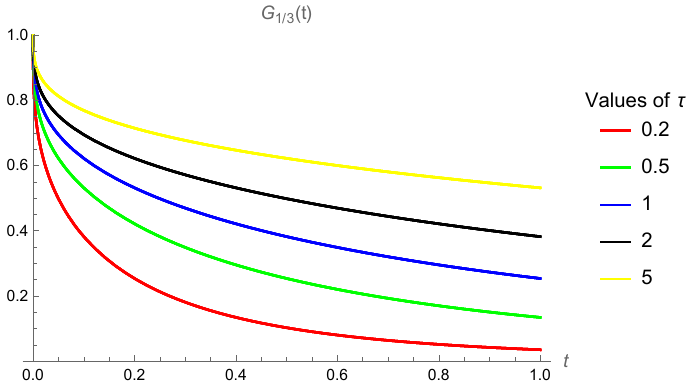}
\caption{Graph of $G_{1/3}\left( t\right) $ for $\protect\mu =1$. }
\label{Figure: G_tercio_plot}
\end{figure}

\begin{figure}[htbp]
\includegraphics[width=10cm]{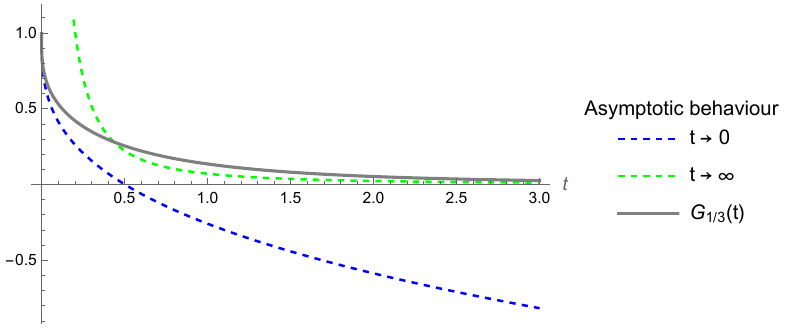}
\caption{Asymptotic 
 behaviour of $G_{1/3}\left( t\right) $ for $\protect\mu %
=1$ and $\protect\tau =\frac{1}{2}$.}
\label{Figure: G tercio asymp}
\end{figure}

\section{Numerical Results} \label{Section: Numerical}

\subsection{Volterra Integral Equation}

The quadrature formulas\ to numerically solve the Volterra integral equation
\cite{Rakotch} that satisfies $G_{\alpha }\left( t\right) $, i.e., Equation 
 (\ref%
{Volterra_Ga(t)}), fail because from (\ref{Asymp_G_t->0}), recalling that $%
\alpha \in \left( 0,1\right) $ and $\mu ,\tau >0$, we have that%
\begin{equation}
\lim_{t\rightarrow 0^{+}}\frac{dG_{\alpha }\left( t\right) }{dt}=-\infty .
\label{Derivative_inf}
\end{equation}
However, we can apply a successive approximation method in order to
numerically compute $G_{\alpha }\left( t\right) $ (\cite{Tricomi} [Sect. 2.1%MDPI: Please check if this should be Section 2.1, and please check if this is the citation of Section 2.1 of this paper <- CONFIRMED
]).
This method states that if we have the Volterra integral equation of the
second kind
\begin{equation}
u\left( t\right) =f\left( t\right) +\int_{0}^{t}K\left( x,t\right) \,u\left(
x\right) \,dx,  \label{Volterra_general}
\end{equation}%
{we} take as zeroth approximation
\begin{equation}
u^{\left( 0\right) }\left( t\right) =f\left( t\right) ,
\label{Volterra_initial_approx}
\end{equation}%
{and} for the successive approximations $j=1,2,\ldots $
\begin{equation}
u^{\left( j\right) }\left( t\right) =f\left( t\right) +\int_{0}^{t}K\left(
x,t\right) \,u^{\left( j-1\right) }\left( x\right) \,dx.
\label{Volterra_succesive_approx}
\end{equation}

Figure \ref{Figure: Volterra}\ shows the application of this successive
approximation method to the solution of (\ref{Volterra_Ga(t)}) (i.e., taking as
kernel $K\left( x,t\right) =-\frac{1}{\tau }\left[ 1+\alpha \left( \frac{t-x%
}{\tau }\right) ^{\alpha -1}\right] $, and $f\left( t\right) =\mu $ in (\ref%
{Volterra_general})), for $\mu =\tau =1$ and $\alpha =1/3$. It is apparent
that as the order of approximation increases, we get a better approximation
to the analytical solution $G_{1/3}\left( t\right) $ obtained in (\ref%
{G_tercio_MR})\ or (\ref{GR_1/3(t)_resultado}). Similar graphs are obtained
for other rational values of $\alpha \in \left( 0,1\right) $ compared to the
analytical solution $G_{\alpha }\left( t\right) $ obtained in (\ref%
{G_m/n(t)_resultado}).

\begin{figure}[htbp]
\includegraphics[width=10cm]{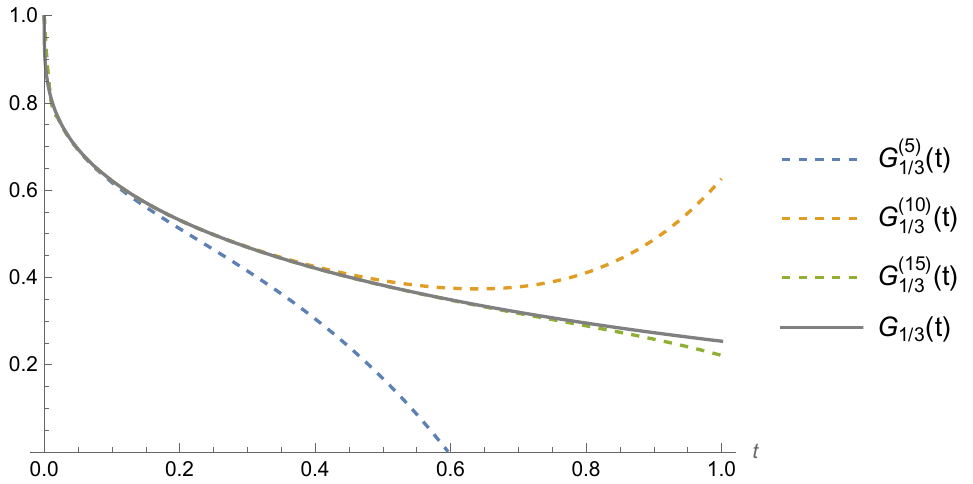}
\caption{Successive approximations $G_{1/3}^{\left( j\right) }\left(
t\right) $ for $\protect\mu =\protect\tau =1$.}
\label{Figure: Volterra}
\end{figure}

Note that this successive approximation method has been successfully applied
in Section \ref{Section: Asymptotic}\ in order to derive the first order
asymptotic formula (\ref{G_a(t)_t->0}).

\subsection{Inverse Laplace Transform}

According to our numerical experiments, the relative error between the
analytical formulas of $G_{\alpha }\left( t\right) $ and the numerical
Laplace inversion of $\tilde{G}_{\alpha }\left( s\right) $, never exceeds
the value of $10^{-9}$ in the time interval $t\in \left[ 0,5\right] $. Below
we present some of these numerical experiments.

Figure \ref{Figure: Delta_MR} shows the relative error $\Delta _{\mathrm{MR}%
}\left( t\right) $ between $G_{1/3}^{\mathrm{MR}}\left( t\right) $ and $%
G_{1/3}^{\mathrm{num}}\left( t\right) $, i.e., the numerical Laplace
inversion of $\tilde{G}_{1/3}\left( s\right) $ using Talbot's method \cite%
{Talbot},
\begin{equation}
\Delta _{\mathrm{MR}}\left( t\right) =\left\vert 1-\frac{G_{1/3}^{\mathrm{MR}%
}\left( t\right) }{G_{1/3}^{\mathrm{num}}\left( t\right) }\right\vert .
\label{Delta_MR_def}
\end{equation}

\begin{figure}[htbp]
\includegraphics[width=8cm]{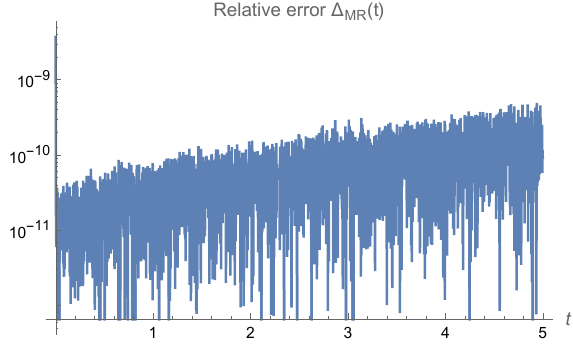}
\caption{Graph of $\Delta _{\mathrm{MR}}\left( t\right) $ for $\protect\mu =%
\protect\tau =1$.}
\label{Figure: Delta_MR}
\end{figure}

Figure \ref{Figure: Delta_R} shows the relative error $\Delta _{\mathrm{R}%
}\left( t\right) $ between $G_{1/3}^{\mathrm{R}}\left( t\right) $ and $%
G_{1/3}^{\mathrm{num}}\left( t\right) $,
\begin{equation}
\Delta _{\mathrm{R}}\left( t\right) =\left\vert 1-\frac{G_{1/3}^{\mathrm{R}%
}\left( t\right) }{G_{1/3}^{\mathrm{num}}\left( t\right) }\right\vert .
\label{Delta_R_def}
\end{equation}

\begin{figure}[htbp]
\includegraphics[width=8cm]{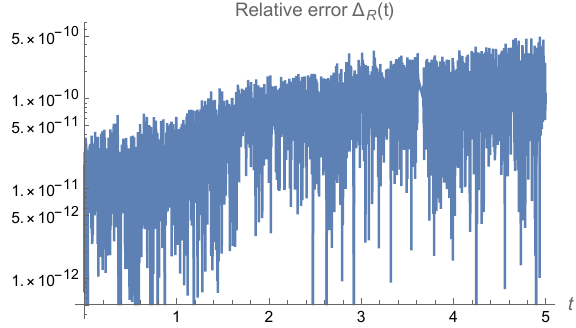}
\caption{Graph %MDPI: Please remove dot after the number. <- I DO NOT UNDERSTAND THIS CORRECTION
 of $\Delta _{\mathrm{R}}\left( t\right) $ for $\protect\mu =%
\protect\tau =1$.}
\label{Figure: Delta_R}
\end{figure}

\section{Conclusions} \label{Section: Conclusions}

Considering the Andrade model in linear viscoelasticity, we have derived for
the first time an analytical expression for the relaxation modulus in the
time domain $G_{\alpha }\left( t\right) $ considering a rational parameter $%
\alpha \in \left( 0,1\right) $ in terms of Mittag--Leffler functions (or
equivalently, as a linear combination of Rabotnov functions). For the
original parameter $\alpha =1/3$ of the Andrade model, we have derived a
particular expression for $G_{1/3}\left( t\right) $ in terms of Miller-Ross
functions. It turns out that this last expression is numerically more
efficient (approximately twice faster) than the equivalent one in terms of
Rabotnov functions.

Furthermore, we have obtained the asymptotic behaviour of $G_{\alpha }\left(
t\right) $ for $t\rightarrow 0^{+}$ and $t\rightarrow +\infty $ using the
Tauberian theorem. We have derived the same expression for the asymptotic
behaviour as $t\rightarrow 0^{+}$ by using the Volterra integral equation of
the second kind that \mbox{$G_{\alpha }\left( t\right) $ satisfies.}

Finally, numerical computations for particular values of the parameters have
been performed in order to verify the analytical solutions obtained. For
this purpose, we have used Talbot's method for the numerical computation of
the inverse Laplace transform, and the method of successive approximations
for the numerical evaluation of the Volterra integral equation of the second
kind.

\end{document}